\documentclass[a4paper,UKenglish, cleveref,autoref, thm-restate]{lipics-v2021}

\pdfoutput=1 %
\hideLIPIcs  %

\title{Resourceful traces for commuting processes}

\author{Matthew Earnshaw}{Tallinn University of Technology, Estonia}{matt@earnshaw.org.uk}{}{Estonian Research Council grant PRG1210.}%

\author{Chad Nester}{University of Tartu, Estonia}{chad.nester@gmail.com}{}{Estonian Research Council grant PRG2764.}

\author{Mario Román}{University of Oxford, United Kingdom}{mromang08@gmail.com}{}{Air Force Office of Scientific Research (AFOSR) award number FA9550-21-1-0038; Advanced Research + Invention Agency (ARIA), Safeguarded AI Programme.}

\authorrunning{M. Earnshaw, C. Nester and M. Román} %

\Copyright{Matthew Earnshaw, Chad Nester, and Mario Román} %



\category{} %

\relatedversion{} %

\nolinenumbers %

\EventEditors{John Q. Open and Joan R. Access}
\EventNoEds{2}
\EventLongTitle{42nd Conference on Very Important Topics (CVIT 2016)}
\EventShortTitle{CVIT 2016}
\EventAcronym{CVIT}
\EventYear{2016}
\EventDate{December 24--27, 2016}
\EventLocation{Little Whinging, United Kingdom}
\EventLogo{}
\SeriesVolume{42}
\ArticleNo{23}

\usepackage{matt-defaults}
\usepackage{float}
\usepackage[utf8]{inputenc}
\usepackage{cmll}

\makeatletter
\newcommand{\oset}[3][0ex]{%
  \mathrel{\mathop{#3}\limits^{
    \vbox to#1{\kern-2\ex@
    \hbox{$\scriptstyle#2$}\vss}}}}
\makeatother

\newcommand\scalemath[3]{\scalebox{#1}[#2]{\mbox{\ensuremath{\displaystyle #3}}}}

\newcommand{\leftarrowtip}{\ensuremath{\tikz\draw[line width=0.5pt,->] (10pt,0) -- (0,0);}}
\newcommand{\leftarrowtailnotip}{\ensuremath{\tikz\draw[line width=0.5pt,-<] (0,0) -- (10pt,0);}}

\newcommand{\unicodeStar}{\ensuremath{\star}}
\DeclareUnicodeCharacter{2605}{\unicodeStar}

\DeclareUnicodeCharacter{21D2}{\ensuremath{\Rightarrow}}
\DeclareUnicodeCharacter{2218}{\ensuremath{\circ}}
\DeclareUnicodeCharacter{2022}{\ensuremath{\bullet}}
\DeclareUnicodeCharacter{2219}{\ensuremath{\bullet}}
\DeclareUnicodeCharacter{2026}{\ensuremath{\dots}}
\DeclareUnicodeCharacter{2208}{\ensuremath{\in}}
\DeclareUnicodeCharacter{2192}{\ensuremath{\to}}
\DeclareUnicodeCharacter{2190}{\ensuremath{\leftarrowtip}}
\DeclareUnicodeCharacter{2919}{\ensuremath{\leftarrowtailnotip}}
\DeclareUnicodeCharacter{2270}{\ensuremath{\nleq}}
\DeclareUnicodeCharacter{2B47}{\ensuremath{\oset{\backsim}{\to}}}
\DeclareUnicodeCharacter{2291}{\ensuremath{\sqsubseteq}}
\DeclareUnicodeCharacter{22A0}{\ensuremath{\boxtimes}}
\DeclareUnicodeCharacter{22B0}{\ensuremath{\boxtimes}}
\DeclareUnicodeCharacter{1D4B9}{\ensuremath{\mathcal{d}}}
\DeclareUnicodeCharacter{1D4C8}{\ensuremath{\mathcal{s}}}

\DeclareUnicodeCharacter{00D7}{\ensuremath{\times}}
\DeclareUnicodeCharacter{00B7}{\ensuremath{\cdot}}
\DeclareUnicodeCharacter{222B}{\ensuremath{\int}}
\DeclareUnicodeCharacter{22A4}{\ensuremath{\top}}
\DeclareUnicodeCharacter{22A5}{\ensuremath{\bot}}
\DeclareUnicodeCharacter{2264}{\ensuremath{\leq}}
\DeclareUnicodeCharacter{2225}{\ensuremath{\parallel}}

\newcommand{\unicodecolon}{\ensuremath{\colon}}
\DeclareUnicodeCharacter{FE55}{\unicodecolon}
\newcommand{\unicodeleftpar}{\ensuremath{\left(}}
\DeclareUnicodeCharacter{27EE}{\unicodeleftpar}
\newcommand{\unicoderightpar}{\ensuremath{\right)}}
\DeclareUnicodeCharacter{27EF}{\unicoderightpar}
\DeclareUnicodeCharacter{2260}{\neq}
\DeclareUnicodeCharacter{22A9}{\Vdash}
\DeclareUnicodeCharacter{2237}{\proportion}
\DeclareUnicodeCharacter{2124}{\mathbb{Z}}
\DeclareUnicodeCharacter{27E8}{\langle}
\DeclareUnicodeCharacter{27E9}{\rangle}
\DeclareUnicodeCharacter{21A6}{\mapsto}
\DeclareUnicodeCharacter{22A2}{\vdash}
\DeclareUnicodeCharacter{2090}{\ensuremath{{}_a}}
\DeclareUnicodeCharacter{A71B}{{}^\uparrow}
\DeclareUnicodeCharacter{A71C}{{}^\downarrow}
\DeclareUnicodeCharacter{27E6}{\llbracket}
\DeclareUnicodeCharacter{27E7}{\rrbracket}

\newcommand{\unicoderightcircle}{\ensuremath{\RIGHTcircle}}
\DeclareUnicodeCharacter{25D1}{\unicoderightcircle}
\newcommand{\unicodeleftcircle}{\ensuremath{\LEFTcircle}}
\DeclareUnicodeCharacter{25D0}{\unicodeleftcircle}
\DeclareUnicodeCharacter{229B}{\circledast}

\newcommand{\unicodebbA}{\ensuremath{\mathbb{A}}}
\DeclareUnicodeCharacter{1D538}{\unicodebbA}
\newcommand{\unicodebbB}{\ensuremath{\mathbb{B}}}
\DeclareUnicodeCharacter{1D539}{\unicodebbB}
\newcommand{\unicodebbC}{\ensuremath{\mathbb{C}}}
\DeclareUnicodeCharacter{2102}{\unicodebbC}
\DeclareUnicodeCharacter{1D53B}{\ensuremath{\mathbb{D}}}
\DeclareUnicodeCharacter{2113}{\ensuremath{\ell}}
\DeclareUnicodeCharacter{2115}{\ensuremath{\mathbb{N}}}
\DeclareUnicodeCharacter{211D}{\ensuremath{\mathbb{R}}}
\DeclareUnicodeCharacter{1D543}{\ensuremath{\mathbb{L}}}
\newcommand\UnicodeBlackboardP{\ensuremath{\mathbf{P}}} \DeclareUnicodeCharacter{2119}{\UnicodeBlackboardP}
\DeclareUnicodeCharacter{211A}{\ensuremath{\mathbb{Q}}}
\DeclareUnicodeCharacter{1D544}{\ensuremath{\mathbb{M}}}
\DeclareUnicodeCharacter{1D54C}{\ensuremath{\mathbb{U}}}
\DeclareUnicodeCharacter{1D54D}{\ensuremath{\mathbb{V}}}
\DeclareUnicodeCharacter{1D54E}{\ensuremath{\mathbb{W}}}
\DeclareUnicodeCharacter{1D546}{\ensuremath{\mathbb{O}}}
\DeclareUnicodeCharacter{1D540}{\ensuremath{\mathbb{I}}}
\DeclareUnicodeCharacter{1D54A}{\ensuremath{\mathbb{S}}}
\DeclareUnicodeCharacter{1D53C}{\ensuremath{\mathbb{E}}}

\newcommand{\unicodecalS}{\ensuremath{\mathcal{S}}}
\newcommand{\unicodecalT}{\ensuremath{\mathcal{T}}}
\newcommand{\unicodecalC}{\ensuremath{\mathcal{C}}}
\newcommand{\unicodecalD}{\ensuremath{\mathcal{D}}}
\newcommand{\unicodecalX}{\ensuremath{\mathcal{X}}}
\newcommand{\unicodecalN}{\ensuremath{\mathcal{N}}}
\newcommand{\unicodecalE}{\ensuremath{\mathcal{E}}}
\DeclareUnicodeCharacter{1D4D4}{\unicodecalE}
\DeclareUnicodeCharacter{1D4D2}{\unicodecalC}
\DeclareUnicodeCharacter{1D4D3}{\unicodecalD}
\DeclareUnicodeCharacter{1D4DE}{\mathcal{O}}
\DeclareUnicodeCharacter{1D4E2}{\unicodecalS}
\DeclareUnicodeCharacter{1D4E3}{\unicodecalT}
\DeclareUnicodeCharacter{1D4E7}{\unicodecalX}
\DeclareUnicodeCharacter{1D4D0}{\ensuremath{\mathcal{A}}}
\DeclareUnicodeCharacter{1D4D1}{\ensuremath{\mathcal{B}}}
\DeclareUnicodeCharacter{1D4D6}{\ensuremath{\mathcal{G}}}
\DeclareUnicodeCharacter{1D4D7}{\ensuremath{\mathcal{H}}}
\DeclareUnicodeCharacter{1D4DB}{\ensuremath{\mathcal{L}}}
\DeclareUnicodeCharacter{1D4DC}{\ensuremath{\mathcal{M}}}
\DeclareUnicodeCharacter{1D4DD}{\unicodecalN}
\DeclareUnicodeCharacter{1D4B1}{\ensuremath{\mathcal{V}}}
\DeclareUnicodeCharacter{1D4E5}{\ensuremath{\mathcal{V}}}
\DeclareUnicodeCharacter{1D4E6}{\ensuremath{\mathcal{W}}}
\DeclareUnicodeCharacter{1D4E4}{\ensuremath{\mathcal{U}}}

\DeclareUnicodeCharacter{03B1}{\alpha}
\DeclareUnicodeCharacter{03B2}{\beta}
\DeclareUnicodeCharacter{03BC}{\mu}
\DeclareUnicodeCharacter{03B4}{\delta}
\DeclareUnicodeCharacter{03B5}{\varepsilon}
\DeclareUnicodeCharacter{03B7}{\eta}
\DeclareUnicodeCharacter{03BB}{\lambda}
\DeclareUnicodeCharacter{03C1}{\rho}
\DeclareUnicodeCharacter{03C8}{\psi}
\DeclareUnicodeCharacter{03C4}{\tau}
\DeclareUnicodeCharacter{03A8}{\Psi}
\DeclareUnicodeCharacter{03C3}{\sigma}
\DeclareUnicodeCharacter{03C6}{\varphi}
\DeclareUnicodeCharacter{03A6}{\Phi}
\DeclareUnicodeCharacter{03A3}{\Sigma}
\DeclareUnicodeCharacter{03D5}{\phi}
\DeclareUnicodeCharacter{03B8}{\theta}
\DeclareUnicodeCharacter{03C0}{\ensuremath{\pi}}
\DeclareUnicodeCharacter{0393}{\Gamma}
\DeclareUnicodeCharacter{0394}{\Delta}
\DeclareUnicodeCharacter{03BA}{\kappa}
\DeclareUnicodeCharacter{03BD}{\nu}
\DeclareUnicodeCharacter{25A0}{\blacksquare}
\DeclareUnicodeCharacter{25AA}{\blacksquare}

\newcommand{\hirayo}{\scaleobj{0.9}{\text{\usefont{U}{min}{m}{n}\symbol{'210}}}}
\DeclareUnicodeCharacter{3088}{\hirayo}
\DeclareFontFamily{U}{min}{}
\DeclareFontShape{U}{min}{m}{n}{<-> udmj30}{}

\newcommand\UnicodeWhiteRightPointingSmallTriangle{\triangleright}
\DeclareUnicodeCharacter{25B9}{\mathbin{\UnicodeWhiteRightPointingSmallTriangle}}
\newcommand\UnicodeWhiteDownPointingSmallTriangle{\triangledown}
\DeclareUnicodeCharacter{25BF}{\mathbin{\UnicodeWhiteDownPointingSmallTriangle}}
\newcommand\UnicodeWhiteUpPointingSmallTriangle{\scalemath{1}{-1}{{}^{\triangledown}}}
\DeclareUnicodeCharacter{25B5}{\mathbin{\UnicodeWhiteUpPointingSmallTriangle}}

\DeclareUnicodeCharacter{2080}{\ensuremath{{}_0}}
\DeclareUnicodeCharacter{2081}{\ensuremath{{}_1}}
\DeclareUnicodeCharacter{2082}{\ensuremath{{}_2}}
\DeclareUnicodeCharacter{2083}{\ensuremath{{}_3}}

\DeclareUnicodeCharacter{1D62}{\ensuremath{{}_i}}
\DeclareUnicodeCharacter{2C7C}{\ensuremath{{}_j}}
\DeclareUnicodeCharacter{02B3}{\ensuremath{{}^r}}
\DeclareUnicodeCharacter{02E1}{\ensuremath{{}^\ell}}
\DeclareUnicodeCharacter{1D48}{\ensuremath{{}^d}}
\DeclareUnicodeCharacter{1D50}{\ensuremath{{}^m}}
\DeclareUnicodeCharacter{1D58}{\ensuremath{{}^u}}
\DeclareUnicodeCharacter{209A}{\ensuremath{{}_p}}
\DeclareUnicodeCharacter{2096}{\ensuremath{{}_k}}
\DeclareUnicodeCharacter{209C}{\ensuremath{{}_t}}

\DeclareUnicodeCharacter{2245}{\ensuremath{\cong}}
\DeclareUnicodeCharacter{2286}{\subseteq}

\DeclareUnicodeCharacter{22C5}{\cdot}
\DeclareUnicodeCharacter{25C3}{\ensuremath{\triangleleft}}
\DeclareUnicodeCharacter{25B9}{\ensuremath{\triangleright}}

\newcommand\smallmath[2]{#1{\raisebox{\dimexpr \fontdimen 22 \textfont 2
      - \fontdimen 22 \scriptscriptfont 2 \relax}{$\scriptscriptstyle #2$}}}
\newcommand\smalloplus{\smallmath\mathbin\oplus}
\newcommand\smallotimes{\smallmath\mathbin\otimes}

\DeclareUnicodeCharacter{2295}{\smalloplus}
\DeclareUnicodeCharacter{2297}{\otimes}
\DeclareUnicodeCharacter{214B}{\parr}
\DeclareUnicodeCharacter{2298}{\oslash}
\DeclareUnicodeCharacter{25C0}{\mathbin{\blacktriangleleft}}
\DeclareUnicodeCharacter{25C1}{\mathbin{\vartriangleleft}}
\DeclareUnicodeCharacter{22B3}{\mathbin{\triangleright}}
\DeclareUnicodeCharacter{22B2}{\mathbin{\triangleleft}}
\DeclareUnicodeCharacter{FF5C}{\mid}
\DeclareUnicodeCharacter{227A}{\mathbin{\prec}}
\DeclareUnicodeCharacter{227B}{\mathbin{\succ}}
\DeclareUnicodeCharacter{22A3}{\mathbin{\dashv}}
\DeclareUnicodeCharacter{219D}{\ensuremath{\leadsto}}
\DeclareUnicodeCharacter{1361}{\colon}

\DeclareUnicodeCharacter{29D1}{\mathrel{\multimapdotbothB}}
\DeclareUnicodeCharacter{29D2}{\mathrel{\multimapdotbothA}}
\DeclareUnicodeCharacter{22C4}{\mathbin{\diamond}}
\DeclareUnicodeCharacter{226B}{\mathrel{\gg}}
\DeclareUnicodeCharacter{25A1}{\Box}
\DeclareUnicodeCharacter{266F}{\sharp}

\DeclareUnicodeCharacter{2190}{\gets}

\DeclareUnicodeCharacter{2099}{_n}
\DeclareUnicodeCharacter{2098}{_m}
\DeclareUnicodeCharacter{1D5AD}{\ensuremath{\mathsf{N}}}

\DeclareUnicodeCharacter{1D4DF}{\mathcal{P}}

\DeclareUnicodeCharacter{2026}{\mydots}
\DeclareUnicodeCharacter{226B}{\gg}
\DeclareUnicodeCharacter{2248}{\UnicodeApprox}
\DeclareUnicodeCharacter{227C}{\preceq}
\newcommand{\UnicodeApprox}{\ensuremath{\approx}}

\usepackage{stmaryrd}
\newcommand{\unicodeRelationalComposition}{\fatsemi}
\DeclareUnicodeCharacter{2A3E}{\unicodeRelationalComposition}

\DeclareUnicodeCharacter{22CA}{\rtimes}
\DeclareUnicodeCharacter{22C9}{\ltimes} %
\definecolor{nordnight}{HTML}{4c566a}
\knowledgestyle{notion}{color=nordnight}

\knowledge{notion}
| \perp
| orthogonal
| orthogonality

\knowledge{notion}
| presentation of an effectful category by devices
| device presentation
| device presentations

\knowledge{notion}
| equation
| equations

\knowledge{notion}
| finitely-supported

\knowledge{notion}
| strict monoidal functor

\knowledge{notion}
| strict premonoidal category
| strict premonoidal categories
| Strict premonoidal categories
| premonoidal categories
| premonoidal category
| premonoidal
| Premonoidal categories

\knowledge{notion}
| strict premonoidal functor
| strict premonoidal functors

\knowledge{notion}
| device graph
| device graphs
| Device graph

\knowledge{notion}
| morphism of device graphs
| device graph morphism
| device graph morphisms

\knowledge{notion}
| distributions
| distribution
| distributed alphabet

\knowledge{notion}
| free monoidal category
| Free strict monoidal category

\knowledge{notion}
| trace monoid
| free partially commutative monoid
| trace monoids

\knowledge{notion}
| independence relation
| independence relations

\knowledge{notion}
| dependency relation
| dependency relations

\knowledge{notion}
| dependency graph
| dependency graphs

\knowledge{notion}
| monoidal graph
| monoidal graphs
| Monoidal graph

\knowledge{notion}
| trace language

\knowledge{notion}
| effectful graph
| Effectful graphs
| Effectful graph
| effectful graphs

\knowledge{notion}
 | effectful category
 | Effectful categories
 | effectful categories

 \knowledge{notion}
 | morphism of effectful graphs

\knowledge{notion}
| morphism of (strict) effectful categories
| morphism of effectful categories

\knowledge{notion}
 | morphism of monoidal graphs
 | morphisms of monoidal graphs

\knowledge{notion}
| strict monoidal category
  | monoidal category
  | monoidal categories

\knowledge{notion}
 | free premonoidal category

 \knowledge{notion}
 | \parallel

 \knowledge{notion}
 | interchange
 | interchanges

 \knowledge{notion}
 | interfere

\knowledge{notion}
 | interference graph

 \knowledge{notion}
 | central

 \knowledge{notion}
 | underlying morphism of monoidal graphs

\knowledge{ignore}
| center
 | \MonGraph
 | graph
 | base
 | devices

\knowledge{notion}
 | \lightning 

\knowledge{notion}
| Mazurkiewicz trace
| trace

\NewDocumentCommand\dev{o}
{
  \IfNoValueTF{#1}
  {\mathsf{dev}}
  {\mathsf{dev}_{#1}}
}
\newcommand\Dev{\mathsf{DevGraph}}

\newcommand\PreMon{\mathsf{PremonCat}}
\newcommand\MonGraph{\mathsf{MonGraph}}
\newcommand\MonCat{\mathsf{MonCat}}
\newcommand\Eff{\mathsf{EffCat}}
\newcommand\EffGraph{\mathsf{EffGraph}}
\newcommand\ef[1]{\mathcal{#1}}
\newcommand\efg[3]{\mathcal{#1} : #2 \to #3}
\newcommand\efc[3]{\mathscr{#1} : #2 \to #3}
\newcommand\efC[1]{\mathscr{#1}}

\newcommand{\Dep}{\mathsf{Dep}}
\newcommand{\Dist}{\mathsf{Dist}}
\newcommand{\C}{\mathbb{C}}
\newcommand{\D}{\mathbb{D}}
\newcommand{\X}{\mathbb{X}}
\newcommand{\Y}{\mathbb{Y}}
\newcommand{\V}{\mathbb{V}}
\newcommand{\W}{\mathbb{W}}

\newcommand{\dv}[1]{\mathcal{D}_{#1}}
\newcommand{\nperp}{\mathbin{\rotatebox{90}{\ensuremath{\nvdash}}}}
\newcommand\freemon[1]{\mathcal{F}_{\smallotimes}#1}

\NewDocumentCommand\freeeff{o}
{
  \IfNoValueTF{#1}
  {\mathscr{F}}
  {\mathscr{F}(#1)}
}

\begin{document}

\maketitle
\begin{abstract}
  We show that, when the actions of a Mazurkiewicz trace are considered not merely as atomic (i.e., mere names) but transformations from a specified type of inputs to a specified type of outputs, we obtain a novel notion of presentation for effectful categories (also known as generalised Freyd categories), a well-known algebraic structure in the semantics of side-effecting computation. Like the usual representation of traces as graphs, our notion of presentation gives rise to a graphical calculus for effectful categories. We use our presentations to give a construction of the commuting tensor product of free effectful categories, capturing the combination of systems in which the actions of each must commute with one another, while still permitting exchange of resources.
\end{abstract} 

\section{Introduction} \label{sec:intro}

Mazurkiewicz traces \cite{booktraces,mazurkiewicz} provide a simple but powerful model of concurrent systems. Traces are a generalization of words, in which specified pairs of symbols (thought of as actions) can commute. Commuting actions $a$ and $b$ are \emph{independent}: their possible concurrent execution $ab$ is observationally indistinguishable from $ba$.

Just as free monoids are the algebraic structure formed by words, \emph{free partially commutative monoids} \cite{foata1969problemes} or \emph{trace monoids} \cite{mazurkiewicz,zielonka} are the algebraic structure formed by (Mazurkiewicz) traces. Traces therefore permit an algebraic approach to the analysis of concurrent systems, by analogy with the use of algebraic methods in automata theory.

Mazurkiewicz states that intuitively, ``\emph{actions} are state transformations of some resources of a system'' and that ``actions are independent if they act on disjoint sets of resources'' \cite{Mazurkiewicz1977}. We contend that this view conflates two notions of resource that might be present in a concurrent system, and it is the remit of this paper to show that by teasing them apart, we obtain a richer algebraic structure.

On the one hand are shared resources corresponding to definite noun phrases, such as ``\emph{the} database'', ``\emph{the} memory location'', or ``\emph{the} printer''. These indeed give rise to relations of (in)dependence between actions. In this paper, we shall refer to such resources as \emph{devices}. On the other hand are the kinds of resources and their transformations axiomatized by \emph{monoidal categories}, as in Coecke, Fritz and Spekkens \cite{coecke2016mathematical}. These are resources which act like \emph{types}, and generally correspond to indefinite noun phrases such as ``\emph{a} database query'', ``\emph{an} integer'', or ``\emph{}a message''. Given any two such resources, we can always consider their conjunction ``in parallel'', and they do not necessarily lead to relations of (in)dependence between actions.

For example, consider \Cref{fig:running-ex}, which introduces a running example, in our graphical notation. The left-hand side of the figure corresponds to our analogue of the alphabet of actions underlying a trace: we call this an \emph{effectful graph}. In our example, the effectful graph contains a process corresponding to the action of \emph{emitting a document}, with a single output string, and two processes corresponding to the action of \emph{printing} received data on two distinct printers. Note that these actions are no longer \emph{atomic}: they have strings corresponding to resources (depicted as solid strings, in general annotated with particular resources), that here indicate the data to be printed. Crucially, the two printing actions are annotated by distinct \emph{devices}, depicted as patterned strings. Intuitively, a \emph{resourceful trace} is a \emph{string diagram} (directed from a left to a right boundary) built by plugging generators together, while only allowing each device string to appear at most once in each vertical section.

\begin{figure}[h]
  \centering
  \includegraphics[width=\textwidth]{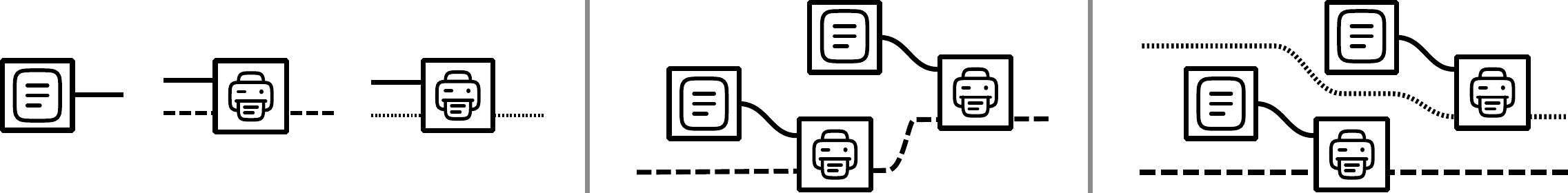}
  \caption{(Left) An effectful graph. (Center) Resourceful trace over the first two elements of the effectful graph, i.e. involving only one printer. (Right) Resourceful trace involving two printers: printings are now independent of one another.}
  \label{fig:running-ex}
\end{figure}

If we only admit \emph{one} printer in our system, printing actions must occur in a specific order, and this is captured by the topology of the trace, as in the center of \Cref{fig:running-ex}. However, if we admit both printers, printing can occur concurrently, and this is again captured by the topology of the traces, such as on the right of \Cref{fig:running-ex}, where the printing actions may slide past one another  or \emph{interchange}. We shall see in \Cref{sec:commutingtensor} that this example in fact arises by a canonical combination of the ``theory of a printer'' with itself, called the \emph{commuting tensor product}, generalizing the commuting tensor product of algebraic theories \cite{Freyd1966}.

The collection of all resourceful traces over an effectful graph assembles into an \emph{effectful category} \cite{roman2025}, also known as generalized Freyd categories \cite{commutativity,10.1007/3-540-48523-6_59}, a well-known algebraic structure in the semantics of effectful computation. Indeed, one of the results of this paper is that resourceful traces are the morphisms of free effectful categories. %

\paragraph*{Related work} That Mazurkiewicz traces can be seen as morphisms in certain free monoidal categories was elaborated by Sobociński and the first author \cite{earnshaw_et_al:LIPIcs.MFCS.2023.43}. By moving to effectful categories, we sharpen this viewpoint, recovering Mazurkiewicz traces precisely as the morphisms of free effectful categories with no resources. The construction of free effectful categories over effectful graphs with a single device was given by the third author, following Jeffrey \cite{jeffrey1997:premonoidal,EPTCS380.20,roman2025}. The construction of the commuting tensor product of effectful categories appeared in \cite{earnshaw2025}, in terms of string diagrams, covering the case of morphisms supported by a finite number of devices. The central results of this paper are new, namely the existence of an adjunction between effectful graphs and effectful categories (\Cref{thm:adjeff}), and \Cref{thm:tensor} on the commuting tensor product of free effectful categories.

\paragraph*{Outline} \Cref{sec:maztrace} recalls the notion of Mazurkiewicz trace, and their graphical representation.  \Cref{sec:free-effectfuls} introduces \emph{effectful graphs} and \emph{effectful categories}, via their underlying \emph{device graphs} and \emph{monoidal graphs}, generalizing the notion of distribution of an alphabet of actions. We give a construction of the free effectful category over an effectful graph, whose morphisms are \emph{resourceful traces}. In \Cref{sec:interference}, we show how the cliques construction from the theory of traces generalizes to effectful categories, and exhibits the free construction as a left adjoint. Finally, \Cref{sec:commutingtensor} gives a construction of the commuting tensor product of free effectful categories, capturing the combination of systems in which the morphisms of each are forced to commute with one another, even while they may exchange resources.
\section{Mazurkiewicz traces} \label{sec:maztrace}

In this section, we recall the basic definitions of trace theory, including the graphical presentation of traces. For more details, we refer to Mazurkiewicz, Hoogeboom and Rozenberg \cite{booktraces}. For this section, we fix a set of actions $\Sigma$.

\begin{definition} 
  A ""dependency relation"", $D ⊆ Σ × Σ$, is a reflexive, symmetric relation. Dependency relations form a preorder, $\Dep_Σ$, with order the inclusion of relations.
\end{definition}

Intuitively, a "dependency relation" specifies actions which should not occur concurrently, because they have some kind of dependency on each other, or the potential to \emph{interfere}, such as writing to the same location in memory.

\begin{definition} \label{defn:tracemonoid}
  Given a "dependency relation" $D \subseteq \Sigma \times \Sigma$, let $\equiv_D$ be the least congruence on $Σ^{*}$ such that for every pair of actions $a,b \in D$, $(a,b) \notin D$ implies $ab \equiv_D ba$. The ""free partially commutative monoid"" or "trace monoid" generated by $D$ is the quotient monoid $\Sigma^{*}/{\equiv_D}$. An element of the trace monoid is a ""Mazurkiewicz trace"" (or simply \emph{trace}) over $(\Sigma, D)$.
\end{definition}

A \emph{trace} is thus an equivalence class of words up to commutation of \emph{independent} actions. Another construction of traces starts from \emph{distributions} of the set of actions:

\begin{definition}%
  \label{defn:dist}
  A ""distribution"" of $\Sigma$ is a function $\mathsf{dev} : \Sigma \to
  \mathscr{P}(\{1,...,k\})$ for some $k \geqslant 1$. Distributions form a
  preorder $\Dist_\Sigma$ with $\mathsf{dev} \geqslant \mathsf{dev}'$ if and
  only if for all $a,b \in \Sigma$,  \[\mathsf{dev}'(a) \cap \mathsf{dev}'(b)
  \neq \varnothing \implies \mathsf{dev}(a) \cap \mathsf{dev}(b) \neq
  \varnothing.\]
\end{definition}

Classically, $\dev[](\sigma)$ is known as the set of \emph{locations} of $\sigma$. In line with the terminology introduced in the following, we call $\mathsf{dev}(\sigma)$ the set of \emph{devices} of $\sigma$. In terms of concurrency, we might consider $\mathsf{dev}(\sigma)$ to be the set of shared resources on which $\sigma$ depends, such as \emph{memory locations}, \emph{execution threads}, \emph{runtimes} or \emph{peripherals}.

We can represent a distribution graphically by introducing nodes (depicted as boxes) corresponding to actions in $\Sigma$, with interfaces (depicted as patterned strings) corresponding to the devices assigned to the action, as for example in \Cref{fig:trace-signature}.

\begin{figure}[h]
  \centering
  \includegraphics[scale=0.65]{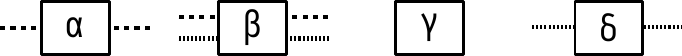}
  \caption{Graphical representation of a distribution: $\alpha$ and $\beta$ share one device, and $\beta$ and $\gamma$ share one device, whereas $\gamma$ has no devices.}
  \label{fig:trace-signature}
\end{figure}

Traces are the diagrams that we can build by plugging these components together, subject to the restriction that each device string occurs exactly once in each vertical slice of the diagram, such as in \Cref{fig:trace}. This is formalized by the algebra of (pre)monoidal categories, as we shall see in the next section. The collection of all such diagrams also forms a monoid, isomorphic to the "trace monoid". More details can be found in the work of Soboci\'{n}ski and the first author \cite{earnshaw_et_al:LIPIcs.MFCS.2023.43}, which builds on earlier work on the representation of traces as graphs \cite{booktraces}.

\begin{figure}[h]
  \centering
  \includegraphics[scale=0.65]{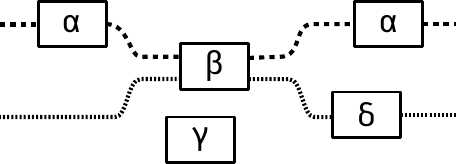}
  \caption{Graphical representation of a trace over the distribution in \Cref{fig:trace-signature}. As an equivalence class of words in the "trace monoid", this trace would include for example $\gamma\alpha\beta\alpha\delta$ and $\alpha\beta\gamma\delta\alpha$: independence is captured by \emph{sliding} actions, while keeping the boundaries fixed.}
  \label{fig:trace}
\end{figure}

A classical construction lets us move between "dependency relations" and "distributions": devices correspond to (non-trivial) \emph{maximal cliques} in the graph of the "dependency relation". For example, the dependency relation depicted in \Cref{fig:depgraph} gives rise to the distribution in \Cref{fig:trace-signature}.  In \Cref{sec:interference}, we shall apply a similar construction to \emph{effectful categories}.

\begin{definition}
  The ""graph"" of a "dependency relation" $\Sigma \subseteq D \times D$ has vertices the elements of $\Sigma$ and an edge $(a,b)$ for every $(a,b) \in D$.
\end{definition}

\begin{figure}[h]
  \centering
  \includegraphics[scale=0.45]{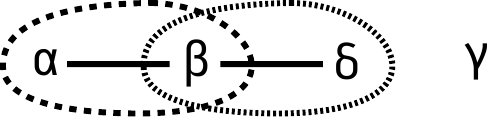}
  \caption{Graph of a dependency relation, whose indicated non-trivial maximal cliques give rise to the distribution of \Cref{fig:trace-signature}.}
  \label{fig:depgraph}
\end{figure}

\begin{proposition}[\S 4.1, \cite{earnshaw_et_al:LIPIcs.MFCS.2023.43}] \label{prop:subposet}
 Every "dependency relation" $D$ on $\Sigma$ gives rise to a "distribution" of $\Sigma$ by taking a bijection of $\{1,...,k\}$ with the non-trivial maximal cliques of the graph of $D$, with $\mathsf{dev}$ sending an action to the subset of such cliques to which it belongs. This extends to a Galois insertion, $\mathsf{cliques} ፡ \Dep_Σ \hookrightarrow \Dist_Σ$.
\end{proposition}

Finally, we remark that any monoid has an underlying "distribution", which can be found using the maximal cliques construction. A monoid $M$ has "dependency relation" $D$ on its elements defined by $(a,b) \in D$ just when $a$ and $b$ do not commute in $M$, and taking non-trivial maximal cliques in the graph of this "dependency relation" defines a distribution of $M$. We shall extend this idea to \emph{effectful categories} in \Cref{sec:interference}.

\section{Resourceful traces and free effectful categories} \label{sec:free-effectfuls}
In traces, actions are merely atomic \emph{names}, equipped with a set of devices. In this section, we enrich the language of traces by allowing actions to additionally have input and output types, which we think of as resources that may be shared between actions. Actions therefore become \emph{morphisms}, with a designated source and target.

In \Cref{sec:effgraphs}, we introduce the resourceful analogue of the "distribution" of a set of actions, which we call an "effectful graph". In an "effectful graph", in addition to a set of devices, every action is equipped with input and output types of resources. Then, in \Cref{sec:free}, we show that just as a "distribution" generates a "free partially commutative monoid", effectful graphs generate free \emph{effectful categories} \cite{roman2025}, also known as \emph{generalized Freyd categories}  \cite{commutativity,10.1007/3-540-48523-6_59}.

\subsection{Effectful graphs} \label{sec:effgraphs}

\kl{Effectful graphs} refine \kl{distributions}: not only by equipping actions with input and output types, but additionally by designating some \emph{central} or \emph{pure} actions, with only the impure actions being potentially able to interfere. In programming language theory, this is the division between \emph{values} and \emph{computations} \cite{POWER2002303,LEVY2003182}.

One reason for this can be seen already at the level of "trace monoids". \emph{Morphisms} of "trace monoids" are, a priori, simply morphisms of monoids, but these need not preserve \emph{centrality} of actions. That is, if an action $\alpha$ commutes with all actions in a "trace monoid" $M$ (i.e., $\alpha$ is \emph{central} in $M$), then its image under a monoid homomorphism $M \to N$ need not commute with all actions of $N$. Thus, certain actions that we might consider to be \emph{pure} might change their role under a translation of systems. However, restricting the notion of morphism to require preservation of centrality is too strong, since it is possible for computations to be incidentally central: an example may be found in Staton and Levy \cite[\S 5.2]{10.1145/2429069.2429091}.

To overcome this, we explicitly designate which actions we want to consider as pure, and ask that centrality of just these actions is preserved. The chosen set of \emph{pure} actions is specified by a \emph{monoidal graph}:

\begin{definition}
  A ""monoidal graph"" $M$ consists of a set of objects, $V_M$; a set of arrows, $E_M$; and a pair of functions $\delta_0,\delta_1 : E_M \to V_M^{*}$ assigning source and target \emph{lists} of objects to each arrow. We can equivalently present a monoidal graph by a set of objects $V_M$, and sets $M(X_1,...,X_n; Y_1,...,Y_m)$ of arrows for each pair of lists of objects: we shall find both perspectives useful.
 \end{definition}

 We can picture a "monoidal graph" as on the left of \Cref{fig:graphs}, where boxes depict the arrows, and their associated strings designate the source and targets. We use \emph{solid} strings for the objects of a "monoidal graph". %

\begin{definition}
  A ""morphism of monoidal graphs"" $\alpha : M \to N$ comprises two functions $\alpha_E : E_M \to E_N$ and $\alpha_V : V_M \to V_N$ commuting with $\delta_0$ and $\delta_1$, that is, $\alpha_V^{*} \circ \delta_i = \delta_i \circ \alpha_E$.

  Equivalently, a morphism is specified by a function $\alpha$ on objects, and functions
  \[M(X_1, ...,X_n; Y_1, ...,Y_m) \to N(\alpha X_1, ..., \alpha X_n; \alpha Y_1 , ..., \alpha Y_m).\]

  It is easy to check that morphisms of monoidal graphs compose, the composite given by composing their underlying actions on arrows and objects, and that we have identities. Denote by $""\MonGraph""$ the category of "monoidal graphs" and their morphisms.
\end{definition}

The set of morphisms that may be equipped with devices, specifying dependencies between morphisms, is given by a \emph{device graph}, which is simply a "monoidal graph" together with extra data associating a set of devices to each morphism:

\begin{definition}
  A ""device graph"" $C$ comprises
  \begin{itemize}
    \item a monoidal graph $|C|$,
    \item a set $\dv{C}$ of devices, and
    \item a function $\dev[C] : E_{|C|} \to \mathcal{P}(\dv{C})$ assigning each arrow in $|C|$ to a subset of the devices. We shall simply write $\dev$ when the "device graph" $C$ is clear from context.
    \end{itemize}
  \end{definition}

  \begin{figure}[h]
    \centering
    \includegraphics[width=0.75\textwidth]{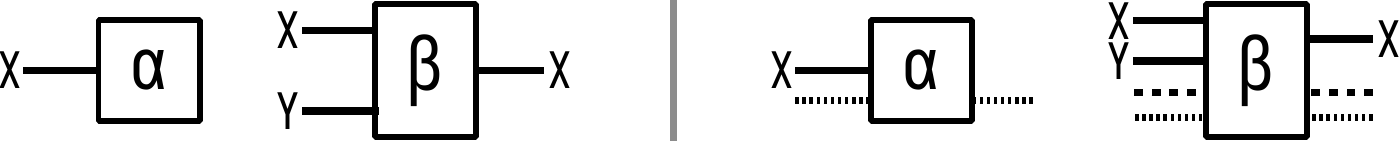}
    \caption{Left: "Monoidal graph" with arrows $\alpha : X \to \varepsilon$, $\beta : XY \to X$. Right: "Device graph" with the same underlying "monoidal graph" as on the left. $\alpha$ and $\beta$ share a device, but $\beta$ has an additional device. The devices of an action all appear in both the source and target, but need not occur in a specific order.}
    \label{fig:graphs}
  \end{figure}

  We can depict a "device graph" as on the right of \Cref{fig:graphs}. Patterned strings denote the devices, which need not occur in a particular order. In contrast to resource strings, device strings must \emph{thread through} the process: this enforces an order between dependent processes.

\begin{example} \label{ex:devmon}
  A "device graph" in which the underlying "monoidal graph" has an empty set of objects and set of arrows $\Sigma$ is precisely a "distribution" of $\Sigma$ in the sense of \Cref{defn:dist}, although we allow the set of devices to be infinite. In practice, a finite set of devices often suffices. For example, if the "device graph" has a finite number of arrows, then an infinite number of devices can be shown to be redundant.%
\end{example}

\begin{definition}
  Morphisms $f$ and $g$ in a "device graph" $C$ are said to be ""orthogonal"", denoted $f \perp_{C} g$ just when they share no devices, $\dev(f) \cap \dev(g) = \varnothing$. Note that any $f$ with no devices is orthogonal to every $g$. We write $f \perp g$ when $C$ is clear from context, and $f \nperp g$ when $f$ and $g$ are not orthogonal.
\end{definition}

A \emph{morphism of device graphs} must preserve "orthogonality",

\begin{definition}
  A ""morphism of device graphs"", $\alpha : C \to D$, comprises
  \begin{itemize}
    \item a morphism of underlying "monoidal graphs", $\alpha : |C| \to |D|$,
    \item such that for all arrows $f,g \in C$, if $f \mathbin{"\perp"_{C}} g$ then $\alpha(f) \mathbin{"\perp"_{D}} \alpha(g)$.
    \end{itemize}
  \end{definition}

\begin{proposition} \label{lem:devgcat}
  Device graphs and their morphisms form a category $\Dev$.
\end{proposition}
\begin{proof}
  Composition of morphisms is given by composition of their underlying "morphisms of monoidal graphs", and identities by the identities on "monoidal graphs". It remains to check that the "orthogonality" condition is satisfied by a composite $\delta \circ \beta : \ef{G} \to \ef{H} \to \ef{K}$.  Indeed, if $f \mathbin{"\perp"_{\ef{G}}} g \Rightarrow \beta(f) \mathbin{"\perp"_{\ef{H}}} \beta(g)$ for all $f,g \in \ef{G}$ and $h \mathbin{"\perp"_{\ef{H}}} i \Rightarrow \delta(h) \mathbin{"\perp"_{\ef{K}}} \delta(i)$ for all $h,i \in \ef{H}$ then we have
  \begin{align*}
    &f \mathbin{"\perp"_{\ef{G}}} g \quad\Rightarrow\quad \beta(f) \mathbin{"\perp"_{\ef{H}}} \beta(g) \quad\Rightarrow\quad \delta(\beta(f)) \mathbin{"\perp"_{\ef{K}}} \delta(\beta(g)) \quad = \quad (\delta\circ \beta)(f) \mathbin{"\perp"_{\ef{K}}} (\delta \circ \beta)(g)
  \end{align*}for all $f,g \in \ef{G}$ as required.
\end{proof}
  
We now have the ingredients required to define \emph{effectful graphs},

\begin{definition}
  An ""effectful graph"" $\ef{G} : V \to C$ comprises
  \begin{itemize}
  \item a "monoidal graph" $V$,
  \item a "device graph" $C$ over the same set of objects as $V$,
  \item an identity-on-objects "morphism of monoidal graphs" $\ef{G} : V \to |C|$,
  \item such that $\dev(\ef{G}v) = \varnothing$ for all arrows $v$ in $V$.
  \end{itemize}
\end{definition}

We usually think of $V$ in terms of its image in $C$, and thus when we speak of a \emph{morphism in an "effectful graph"}, we are referring to the arrows of the "device graph" $C$.

\begin{example}
  In the "effectful graph" in \Cref{fig:running-ex}, the "monoidal graph" comprises solely the ``document'' generator, while the "device graph" includes all three generators, with the "morphism of monoidal graphs" simply including the ``document'' generator.
\end{example}

\begin{definition}
  A ""morphism of effectful graphs"" $(\alpha_0,\alpha) : (\ef{G} : V \to C) \to (\ef{H} : W \to D)$ comprises %
  \begin{itemize}
  \item a "morphism of monoidal graphs" $\alpha_0 : V \to W$, and
  \item a "morphism of device graphs" $\alpha : C \to D$,
  \item such that the square of "morphisms of monoidal graphs" commutes, $|\alpha| \circ \ef{G} = \ef{H} \circ \alpha_0$.
  \end{itemize}
\end{definition}

Since both $\MonGraph$ and $\Dev$ are categories, it follows easily that:

\begin{proposition}
"Effectful graphs" and their morphisms form a category $\EffGraph$.
\end{proposition}

 Note that "effectful graphs" generalize the \emph{effectful polygraphs} of Sobociński and the third author \cite{roman2025}, precisely by allowing each action to have a different \emph{set} of devices, rather than all actions having a single, global device. The construction of effectful categories in the next section also refines that from \emph{effectful polygraphs} \cite{roman2025}, in that it builds-in extra independence equations between processes.

\subsection{Effectful categories}
"Effectful graphs" generate \emph{effectful categories}. An "effectful category" comprises a monoidal category $\mathbb{V}$, a premonoidal category $\C$, and an identity-on-objects premonoidal functor $\mathbb{V} \to \C$. In this paper, we work with \emph{strict} versions of these notions, since not only are the free such structures \emph{strict}, but coherence theorems exist which state that any non-strict (pre)monoidal category is equivalent to a strict one \cite{power97}.

Just as with "effectful graphs", the idea is that $\mathbb{V}$ is a category of pure actions (generated by a "monoidal graph"), $\C$ is a category of effectful actions (generated by a "device graph"), and the functor is thought of as an inclusion of $\mathbb{V}$ into $\C$. Let us build up to their definition.

\begin{definition} \label{defn:strictmon}
  A ""strict monoidal category"" is a category $\V$ equipped with
  \begin{itemize}
  \item a functor $\otimes : \V \times \V \to \V$, called the monoidal product,
  \item a distinguished object $I$, called the monoidal unit,
  \item such that $A \otimes (B \otimes C) = (A \otimes B) \otimes C$ and $A \otimes I = A = I \otimes A$,
  \item and such that the following equations hold for all morphisms $f,g,h$:
     \[f \otimes (g \otimes h) = (f \otimes g) \otimes h \\
      1_I \otimes f = f = f \otimes 1_I\]
  \end{itemize}
\end{definition}

"Premonoidal categories" refine "monoidal categories" by only requiring the monoidal product to be functorial separately in each variable: in particular, there is no operation combining morphisms in parallel.

\begin{definition} \label{defn:strictpre}%
  \AP A ""strict premonoidal category"" is a category, $\C$, equipped with
  functors $(A \ltimes -) : \C \to \C$ (the ``left-whiskering'' by A) and $(-
  \rtimes A) : \C \to \C$ (the ``right-whiskering'' by A), for every object $A$ of
  $\C$, and a distinguished object $I$ (the ``monoidal unit''), such that,
  \begin{itemize}
    \item $A \rtimes B = A \ltimes B$, allowing us to denote this object as $A \otimes B$,
    \item $A \otimes (B \otimes C) = (A \otimes B) \otimes C$ and $A \otimes I = A = I \otimes A$,
    \item and such that the following equations (\Cref{fig:premeq}) hold, stating coherence of left- and right-whiskering with each other and with the monoid structure,
      \begin{figure}[H]
        \vspace*{-5mm}
      \[ (A \ltimes f) \rtimes B = A \ltimes (f \rtimes B) \\ \qquad
        I \ltimes f = f = f \rtimes I \] \vspace*{-7mm}
      \[ (A \otimes B) \ltimes f = A \ltimes (B \ltimes f)  \\ \qquad
       f \rtimes (A \otimes B) = (f \rtimes A) \rtimes B \]
    \caption{Equations for left and right whiskerings in a strict premonoidal category.}
    \label{fig:premeq}
  \end{figure}
  \end{itemize}
\end{definition}

Note that every "monoidal category" is a "premonoidal category" in which partial applications of the monoidal product define the left- and right- whiskerings. The key equations that hold for every pair of morphisms in a "monoidal category", but not generally in a given "premonoidal category", are those expressing \emph{interchange} (i.e., \emph{commutation}) of morphisms:

\begin{definition}
  For morphisms $f : A \to B$ and $g : A' \to B'$ in a "premonoidal category" $\C$,
  \begin{itemize}
  \item we write $f ""\parallel"" g$ in case $(f \rtimes A')(B \ltimes g) = (A \ltimes g)(f \rtimes B')$,
  \item we say that $f$ and $g$ ""interchange"" in case $f \parallel g$ and $g \parallel f$,
  \item we say that $f$ and $g$ ""interfere"", written $f \lightning g$, in case they do not interchange,
  \item we say that $f$ is ""central"" in case it "interchanges" with every morphism of $\C$.
  \end{itemize}
\end{definition}

Graphically, interchange of morphisms is precisely the \emph{sliding} of morphisms past one another, as discussed for example in \Cref{fig:trace}.

\begin{definition}
  Let $\C$ and $\D$ be "strict premonoidal categories". A ""strict premonoidal functor"" is a functor $F : \C \to \D$ that
  \begin{itemize}
  \item maps objects as a monoid homomorphism, $F(A \otimes B) = F(A) \otimes F(B)$ and $F(I) = I$,
  \item and preserves whiskerings, $F(A \ltimes f) = F(A) \ltimes F(f)$ and $F(f \rtimes A) = F(f) \rtimes F(A)$.
  \end{itemize}
  When $\C$ and $\D$ are moreover "strict monoidal categories", we call this a ""strict monoidal functor"".
\end{definition}

Finally, we reach the definition of an "effectful category":

\begin{definition} \label{defn:eff}
  A (strict) ""effectful category"" $\efc{E}{\V}{\C}$ comprises
  \begin{itemize}
    \item a strict "monoidal category" $\V$,
    \item a strict "premonoidal category" $\C$ with $\obj{\C} = \obj{\V}$, and
    \item an identity-on-objects "strict premonoidal functor" $\efc{E}{\V}{\C}$,
    \item such that the image of $\efc{E}{\V}{\C}$ is central.
  \end{itemize}
\end{definition}

\begin{definition}
  A ""morphism of (strict) effectful categories"", $(\alpha_0,\alpha) : (\efc{E}{\V}{\C}) \to (\efc{F}{\W}{\D})$, comprises a "strict monoidal functor" $\alpha_0 : \V \to \W$ and a "strict premonoidal functor" $\alpha : \C \to \D$ such that $\alpha \circ \efC{E} = \efC{F} \circ \alpha_0$ in $\PreMon$ (eliding the inclusion of $\MonCat \hookrightarrow \PreMon$).
\end{definition}

From the fact that $\MonCat$ and $\Dev$ are categories, it follows easily that:

\begin{proposition}
"Effectful categories" and their morphisms form a category $\Eff$.
\end{proposition}

\subsection{Free effectful category over an effectful graph} \label{sec:free}

Our goal now is to define a functor $\freeeff{} : \EffGraph \to \Eff$. We will see in the next section that this functor has a right adjoint, and so truly constructs the \emph{free} "effectful category" on an "effectful graph". We begin by constructing the free "premonoidal category" over a "device graph".

\begin{definition} \label{defn:freepre}
  For a "device graph" $G$, we construct a "premonoidal category" $\mathcal{F}G$ on $G$ as follows. The underlying category has,
  \begin{itemize}
  \item set of objects $V_G^*$, the free monoid on $V_G$, whose elements are lists. We shall denote by $\otimes$ the concatenation of lists.
  \item set of morphisms $E_{\mathcal{F}G}/{\equiv}$, where $E_{\mathcal{F}G}$ is the set inductively defined by rules for identities, the whiskering of generating arrows, and composition, %
    \begin{mathpar}
\inferrule[id]{X \in V_G^*}{1_X : X \to X}

      \inferrule[whisk]{X, Y \in V_G^{*} \\ g : W \to Z \in E_G}{X \triangleright g \triangleleft Y : X \otimes W \otimes Y \to X \otimes Z \otimes Y}

      \inferrule[comp]{u : P \to Q \\ v : Q \to R}{u \comp v : P \to R}

    \end{mathpar}
  \item and $\equiv$ is the least congruence for $\comp$ generated by, %
    \[
      (u \comp v) \comp w = u {\comp} (v \comp w) \\ \qquad
      1_X \comp u = u = u \comp 1_Y %
    \]
  \end{itemize}
  whenever these are well typed, along with equations allowing (whiskerings of) orthogonal generating arrows to interchange, i.e. for every pair, $f \colon U \to V$ and $g \colon U' \to V'$, of arrows in $E_G$ that are orthogonal, $\dev[](f) \cap \dev[](g) = \varnothing$, and each triple of objects $X,Y,Z$ in $V_G^{*}$,
  \begin{flalign*}
        &(X \triangleright f \triangleleft Y \otimes U \otimes Z) \comp (X \otimes V \otimes Y \triangleright g \triangleleft Z)  = \\ &(X \otimes U' \otimes Y \triangleright g \triangleleft Z) \comp (X \triangleright f \triangleleft Y \otimes V' \otimes Z).
  \end{flalign*}

    It is immediate that this data forms a category. For the "premonoidal" structure, we define $I$ to be the empty list, left and right whiskerings act on objects by list concatenation (denoted $\otimes$), and their action on morphisms is defined inductively by
    \[(X \ltimes f) := \begin{cases}
      1_{X \otimes Y} & \text{if $f = 1_Y$,} \\
      (X \ltimes g) {\comp} (X \ltimes h) & \text{if $f = g {\comp} h$,} \\
      X \otimes Y \triangleright f' \triangleleft Z & \text{if $f = Y \triangleright f' \triangleleft Z$,}
    \end{cases}\]
    \[(f \rtimes X) := \begin{cases}
      1_{Y \otimes X} & \text{if $f = 1_Y$,} \\
      (g \rtimes X) {\comp} (h \rtimes X) & \text{if $f = g {\comp} h$,} \\
      Y \triangleright f' \triangleleft Z \otimes X & \text{if $f = Y \triangleright f' \triangleleft Z.$}
    \end{cases}\]

    It is easily verified that these mappings are well-defined, functorial, and satisfy the required coherence equations of \Cref{defn:strictpre}.
\end{definition}

The morphisms of $\mathcal{F}G$ are ""resourceful traces"" over $G$, and can be depicted intuitively as ``string diagrams'' (from an input boundary to an output boundary), built by combining whiskered generators of $G$ via the operations of premonoidal categories. Generators correspond to boxes as in \Cref{fig:graphs}, identity morphisms correspond to strings, composition is given by joining resource strings and any matching device strings, leading to a morphism with the union of the two sets of devices, and whiskering by juxtaposing strings with boxes. Whiskering by the empty list is simply depicted by drawing no whiskering string. We illustrate this in \Cref{fig:restraces}. The requirement that every device string appear at most once in every vertical slice through the diagram, as in \Cref{sec:maztrace}, is enforced by the fact that we have no operation for combining morphisms in parallel. %
The following lemma shows that interchange of generating arrows extends to the morphisms of $\mathcal{F}G$ via the device assignment described above.

\begin{figure}[h]
  \centering
    \includegraphics[width=0.8\textwidth]{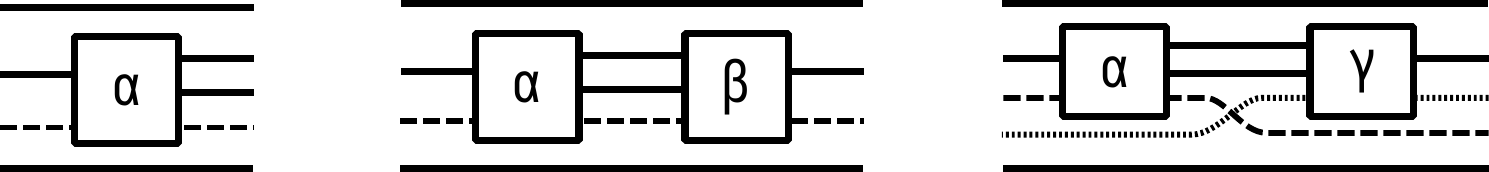}
    \caption{(Left to right): Whiskering of a generating arrow by a resource, composition of morphisms with a matching device, composition of morphisms with distinct devices. In each composition, we obtain a new morphism whose devices are the union of the components. Types of the resource strings are omitted here for clarity.}
    \label{fig:restraces}
  \end{figure}

\begin{restatable}[]{lemma}{interchange} \label{lem:ic}
  Consider the inductive extension of the device assignment of a "device graph" $G$ to morphisms of $\mathcal{F}G$ (\Cref{defn:freepre}),
    \[
      \dev[{\mathcal{F}G}](f) := \begin{cases}
        \varnothing & \text{if $f = 1_A$,}\\
      \dev[G](f') & \text{if $f = X \triangleright f' \triangleleft Y$,}\\
      \dev[{\mathcal{F}G}](g) \cup \dev[{\mathcal{F}G}](h) & \text{if $f = g {\comp} h$.}
      \end{cases}
    \]
Let $f,g$ be two morphisms of $\mathcal{F}G$. If $f$ and $g$ are "orthogonal", $\dev[\mathcal{F}G](f) \cap \dev[\mathcal{F}G](g) = \varnothing$, then $f \parallel g$ and $g \parallel f$.
\end{restatable}
\begin{proof}
  See \Cref{app}.
\end{proof}

\begin{example}
  As noted in \Cref{ex:devmon}, a "device graph" $G$ with empty set of objects is precisely a "distribution" of an alphabet. In this case, $\mathcal{F}G$ has a single object, the empty list, and hence is a monoid (with trivial whiskerings). This monoid is isomorphic to the trace monoid over the corresponding distribution.
\end{example}

\begin{lemma} \label{lem:freepf}
  Each "morphism of device graphs", $\alpha : G \to H$, induces  a "strict premonoidal functor", $\mathcal{F}(\alpha) : \mathcal{F}G \to \mathcal{F}H$, defined as follows. The action on objects is the lifting of $\alpha_V$ to lists, $\alpha_V^* : V_G^* \to V_H^*$. The action on morphisms is given by,
  \[
  \mathcal{F}(\alpha)(f) := \begin{cases}
    1_{\mathcal{F}(\alpha)(X)} & \text{if $f = 1_X$}, \\
    \mathcal{F}(\alpha)(X) \triangleright \alpha_E(f') \triangleleft \mathcal{F}(\alpha)(Y) & \text{if $f = X \triangleright f' \triangleleft Y$} \\
    \mathcal{F}(α)(g) \comp \mathcal{F}(α)(h) & \text{if $f = g \comp h$.}\\
  \end{cases}
\]
and this action is well-defined.
\end{lemma}
\begin{proof}
  Well-definedness follows from the fact $\mathcal{F}G$ and $\mathcal{F}H$ are categories and that $\alpha$ preserves orthogonality by definition. The action on objects is a morphism of monoids, so to see that $\mathcal{F}(\alpha)$ is a "strict premonoidal functor", it remains to check that it preserves whiskerings: this follows by a straightfoward induction on $f$.
\end{proof}

It is clear that the mapping $\mathcal{F}$ preserves composition and identities, and thus,

\begin{proposition}
  The assignments of \Cref{defn:freepre} and \Cref{lem:freepf} define a functor $\mathcal{F} : \Dev \to \mathsf{PreMon}$.
\end{proposition}

To define the free "effectful category" over an "effectful graph", we simply combine the construction of the "free monoidal category" on a "monoidal graph" (\Cref{defn:freemon}) with the construction of the "free premonoidal category" on a "device graph" (\Cref{defn:freepre}), as follows.

\begin{definition} \label{defn:freeeff}
  Let $\efg{G}{V}{C}$ be an "effectful graph". This generates the following "effectful category" $\freeeff{\ef{G}} :  \freemon{V} \to \mathcal{F}C$, which \Cref{prop:freeeff} checks is well-defined.

  \begin{itemize}
  \item The "monoidal category" is the "free monoidal category" $\freemon{V}$ on the "monoidal graph" $V$,
  \item the "premonoidal category" $\mathcal{F}{C}$ is the "free premonoidal category" on the "device graph" $C$,%
  \item$\freeeff{\ef{G}} :  \freemon{V} \to \mathcal{F}C$ is taken to be identity-on-objects, and is defined by structural induction on its arguments, following \Cref{defn:freemon}. The interesting case is, for $g_1 : A_1 \to B_1$, $g_2 : A_2 \to B_2$, \[\freeeff{\ef{G}}(g_1 \otimes g_2) := (\freeeff{\ef{G}}(g_1) \rtimes A_2)(B_1 \ltimes \freeeff{\ef{G}}(g_2)) = (A_1 \ltimes \freeeff{\ef{G}}(g_2))(\freeeff{\ef{G}}(g_1) \rtimes B_2).\]
    For this equality to hold, we must have that $\freeeff{\ef{G}}(g_1)$ and $\freeeff{\ef{G}}(g_2)$ are "orthogonal" for all $g_1$ and $g_2$ in $\freemon{V}$, since then we can conclude from \Cref{lem:ic} that they "interchange". It suffices to show that $\freeeff{\ef{G}}(g)$ has no devices for all $g$ in $\freemon{V}$, which follows by structural induction (c.f. \Cref{defn:freemon}).
  \end{itemize}
\end{definition}

\begin{proposition} \label{prop:freeeff}
  \Cref{defn:freeeff} has the structure of an "effectful category".
\end{proposition}
\begin{proof}
  By the definitions of "device graph", "free monoidal category" and "free premonoidal category", we have that $\obj{(\mathcal{F}C)} = \obj{(\mathcal{F}_{\smallotimes}V)} := \obj{V}^{*}$. We must check that $\freeeff{\ef{G}} :  \freemon{V} \to \mathcal{F}C$ is a "strict premonoidal functor" with central image. $\freeeff{\ef{G}}$ is identity-on-objects and so a monoid homomorphism. Furthermore it preserves left-whiskering since, \[\freeeff{\ef{G}}(A \otimes f) := (\freeeff{\ef{G}}(\id{A}) \rtimes B)(A \ltimes \freeeff{\ef{G}}(f)) = A \rtimes \freeeff{\ef{G}}(f),\] using the equations of "strict premonoidal categories", and preserves right-whiskerings in a similar fashion. It follows from \Cref{lem:ic} that the image of $\freeeff{\ef{G}}$ is "central" since the definition of "device graph" asks that $\dev[\ef{G}](v) := \varnothing$ for all $v \in V$, and so by \Cref{defn:freepre}, $\dev(\freeeff{\ef{G}}(g)) = \varnothing$ for all $g \in \freemon{V}$ .
\end{proof}

\begin{proposition} \label{prop:freeeffm}
  Let $(\alpha_0,\alpha) : (\efg{G}{V}{C}) \to (\efg{H}{W}{D})$ be a "morphism of effectful graphs". This generates a "morphism of effectful categories" $\freeeff{(\alpha_0,\alpha)} : \freeeff{\ef{G}} \to \freeeff{\ef{H}}$ defined as follows:
  \begin{itemize}
  \item the underlying "strict monoidal functor" is $\mathcal{F}_{\smallotimes}\alpha_0 : \mathcal{F}_{\smallotimes}V \to \mathcal{F}_{\smallotimes}W$, %
  \item the underlying "strict premonoidal functor" $\mathcal{F}\alpha : \mathcal{F}C \to \mathcal{F}D$ is that given by \Cref{lem:freepf},
  \item and the required square commutes by functoriality.
  \end{itemize}
\end{proposition}

It is easily checked that this assignment on morphisms preserves identities and composition.

\begin{proposition}
  The assignments of \Cref{prop:freeeff} and \Cref{prop:freeeffm} extend to a functor $\freeeff : \EffGraph \to \Eff$.
\end{proposition}

\section{Interference: from effectful categories to effectful graphs} \label{sec:interference}
In this section we show that every "effectful category" has an underlying "effectful graph", given by a generalization of the clique construction of \Cref{prop:subposet}. Rather than taking the elements of a monoid as the vertices, we take the vertices to be the morphisms of a "premonoidal category".

\begin{definition}
  The ""interference graph"" of a "premonoidal category" $\C$, denoted by ${\lightning}\C$, is the graph whose vertices are the arrows of $\C$, with an edge $(f,g)$ if and only if $f$ and $g$ "interfere" in $\C$ (denoted $f ""\lightning"" g$): that is, they do not "interchange".
\end{definition}

\begin{definition} \label{defn:underlyingdev}
  Let $\C$ be a "strict premonoidal category". We define its underlying "device graph" $\mathcal{U}(\C)$ as follows:
  \begin{itemize}
  \item For the underlying "monoidal graph" we take $V_{\mathcal{U}(\C)} := \obj{\C}$, and for every pair of lists $X_1, ..., X_n$ and $Y_1, ..., Y_m$ of objects of $\obj{\C}$, we take a set of arrows
        \[\mathcal{U}(\C)(X_1, ..., X_n; Y_1, ..., Y_m) := \C(X_1 \otimes ... \otimes X_n; Y_1 \otimes ... \otimes Y_m).\]
    That is, there is an arrow $X_1, \ldots, X_n \to Y_1, \ldots, Y_m$ in $\mathcal{U}(\C)$ for every morphism $X_1 \otimes \ldots \otimes X_n \to Y_1 \otimes \ldots \otimes Y_m$ in $\C$. Since each $X_i$ and $Y_j$ is an object of $\C$, whose set of objects already forms a monoid with operation denoted $\otimes$, a morphism $f : X_1 \otimes ... \otimes X_n \to Y_1 \otimes ... \otimes Y_m$ is included both in $\mathcal{U}(\C)(X_1 \otimes ... \otimes X_n; Y_1 \otimes ... \otimes Y_m)$ and $\mathcal{U}(\C)(X_1, ..., X_n; Y_1, ..., Y_m)$. The complete set of arrows $E_{\mathcal{U}(\C)}$ is the coproduct of all such sets, over every pair of lists of objects in $\C$.

  \item The set of devices, $\mathsf{cliques}({\lightning}\C)$, is given by the set of non-trivial maximal cliques in the "interference graph", ${\lightning}\C$.
    A clique is \emph{trivial} just when it is a singleton. A clique is maximal just in case any morphism that interferes with every morphism of the clique is in the clique. The device assignment, $\dev : E_{\mathcal{U}(\C)} \to \mathscr{P}(\mathsf{cliques}({\lightning}\C))$, assigns an arrow to the subset of non-trivial maximal cliques to which it belongs.
  \end{itemize}
\end{definition}

\begin{lemma}\label{lem:reflect-interference}
  For $F : \X \to \Y$ a "strict premonoidal functor", we have that
  $f \parallel g$ implies $F(f) \parallel F(g)$, and that
  $F(f) \lightning F(g)$ implies $f \lightning g$.
\end{lemma}
\begin{proof}
Immediate.
\end{proof}

\begin{proposition} \label{lem:underlyingdevm}
  For every "strict premonoidal functor" $F : \C \to \D$, there is a "morphism of device graphs" $\mathcal{U}(F) : \mathcal{U}(\C) \to \mathcal{U}(\D)$ given by the action of $F$.%
\end{proposition}
\begin{proof}
  That this defines a morphism of the underlying "monoidal graphs" follows from the fact that a functor preserves sources and targets.
  We show that "orthogonality" of morphisms is preserved by proving the contrapositive, $\mathcal{U}(F)(f) \nperp \mathcal{U}(F)(g) \implies f \nperp g$. Suppose that $C \in \dev[\mathcal{U}(\D)](\mathcal{U}(F)(f)) \cap \dev[\mathcal{U}(\D)](\mathcal{U}(F)(g))$ for some $C \in \mathsf{cliques}({\lightning}\D)$. Then $\mathcal{U}(F)(f) = F(f) \in C$ and $\mathcal{U}(F)(g) = F(g) \in C$, and we have $F(f) \lightning F(g)$. Then \Cref{lem:reflect-interference} gives $f \lightning g$, which means $f,g \in C'$ for some $C' \in \mathsf{cliques}({\lightning}\C)$, in which case $C' \in \dev[\mathcal{U}(\C)](f) \cap \dev[\mathcal{U}(\C)](g)$, that is, $f \nperp g$, which is what we wanted to show. %
\end{proof}

\begin{lemma}
  The assignments of \Cref{defn:underlyingdev,lem:underlyingdevm} extend to a functor $\mathcal{U} : \PreMon \to \Dev$.
\end{lemma}

We can now define a functor $\mathcal{U} : \Eff \to \EffGraph$. We begin with the action on "effectful categories".

\begin{proposition} \label{defn:ueff}
  Let $\efc{E}{\V}{\C}$ be an "effectful category". Then $\mathscr{U}(\efC{E}) : \mathcal{U}_{\smallotimes}(\V) \to \mathcal{U}(\C)$ is the "effectful graph" with
  \begin{itemize}
  \item "monoidal graph" $\mathcal{U}_{\smallotimes}(\V)$ (\Cref{defn:umon}),
  \item "device graph" $\mathcal{U}(\C)$,
  \item "morphism of monoidal graphs" $\mathscr{U}(\efC{E}) : \mathcal{U}_{\smallotimes}(\V) \to | \mathcal{U}(\C) |$ defined to be identity-on-objects, and $f \mapsto \efC{E}(f)$ on arrows.
  \end{itemize}
  To see that this is well-defined, first note that by definition, $\mathcal{U}(\C)$ has the same set of objects as $\mathcal{U}_{\smallotimes}(\V)$. It remains to show that $\dev(\mathscr{U}(\efC{E})(v)) = \varnothing$ for all arrows $v$ in $\mathcal{U}_{\smallotimes}(\V)$. By definition of $\mathscr{U}(\efC{E})$, we have $\mathscr{U}(\efC{E})(v) = \efC{E}(v)$. Also, by definition $\efc{E}{\V}{\C}$ has central image, and so $\efC{E}(v)$ must constitute a trivial clique in ${\lightning}\C$.%
\end{proposition}

It follows straightforwardly by  combining \Cref{defn:umon} and \Cref{defn:ueff}, that every "morphism of effectful categories" has an underlying "morphism of effectful graphs", and furthermore that these assignments extend to a functor:

\begin{proposition} \label{lem:ueffm}
  Let $\efc{G}{\V}{\C}$ and $\efc{H}{\W}{\D}$ be "effectful categories" and $(\alpha_0, \alpha) : (\efc{G}{\V}{\C}) \to (\efc{H}{\W}{\D})$ a "morphism of effectful categories". Then, there is an underlying "morphism of effectful graphs", $\mathscr{U}(\alpha_0, \alpha) := (\mathcal{U}_{\smallotimes}(\alpha_0), \mathcal{U}(\alpha))$, where $\mathcal{U}_{\smallotimes}$ is the "underlying morphism of monoidal graphs".
\end{proposition}

\begin{proposition}
  The assignments of \Cref{defn:ueff,lem:ueffm} extend to a functor\\
  $\mathscr{U} : \Eff \to \EffGraph.$
\end{proposition}

Finally, we show that our free and forgetful functors form an adjunction. Proofs can be found in Appendix \ref{app}.

\begin{restatable}[]{theorem}{adj} \label{thm:adj}
  The free construction of "premonoidal categories" over a "device graph" (\Cref{defn:freepre}) is left adjoint to forming the underlying "device graph" of a "premonoidal category" (\Cref{defn:underlyingdev}). That is, $\mathcal{F} \dashv \mathcal{U} : \PreMon \leftrightarrows \Dev.$
  \end{restatable}

\begin{restatable}[]{corollary}{adjeff}  \label{thm:adjeff}
  The adjunction of \Cref{thm:adj} extends to an adjunction between "effectful graphs" and "effectful categories". That is, $\mathscr{F} \dashv \mathscr{U} : \Eff \leftrightarrows \EffGraph.$
\end{restatable}
\section{Commuting tensor product of free effectful categories} \label{sec:commutingtensor}

In their abstract investigation of the notion of \emph{commutativity}, Garner and López Franco deduce the existence of a \emph{commuting tensor product} of effectful categories \cite[\S 9.4]{commutativity}. In this section, we show how our construction of free "effectful categories" from "effectful graphs" gives rise to a simple, concrete construction of this commuting tensor product.

Roughly speaking, given two free "effectful categories" generated by "effectful graphs" over the same "monoidal graph" of chosen pure actions, their commuting tensor product is given by taking the free "effectful category" over a sort of ``disjoint union'' of the two underlying "effectful graphs". In particular, the set of devices is a disjoint union, and so actions from one graph are forced to interchange with those from the other. However, the set of resources is left unchanged, and so the actions from each graph may share resources when they are composed into resourceful traces, as in \Cref{fig:running-ex} (see also \Cref{ex:printer}, below).

\begin{construction} \label{defn:ctg}
   Given "effectful graphs" $\efg{G}{W}{C}$ and $\efg{H}{W}{D}$ over the same "monoidal graph" $W$, we construct their ""commuting tensor product"" $\efg{G \odot H}{W}{C \odot D}$ as follows. The "monoidal graph" of pure morphisms is simply $W$. The "device graph" $C \odot D$ has,
  \begin{itemize}
  \item objects $V_{C\odot D} := V_W = V_C = V_D$,
  \item arrows $E_{C \odot D} := E_C +_{E_W} E_D$, i.e., the following pushout of sets,
    \[
    \begin{tikzcd}
      E_W \ar[r,"\ef{H}_E"] \ar[d,"\ef{G}_E"'] & E_D \ar[d,"R^{C,D}_E"] \\
      E_C \ar[r,"L^{C,D}_E"'] & E_C +_{E_W} E_D
      \arrow["\lrcorner"{anchor=north, pos=-0.15, rotate=180}, draw=none, from=2-2, to=1-1]
    \end{tikzcd}
    \]
  \item devices $\dv{C \odot D} := \dv{C} + \dv{D}$,
  \item device assignment $\dev[{C\odot D}] : E_C +_{E_W} E_D \to \mathcal{P}(\dv{C \odot D})$ the unique map induced by the universal property of the pushout via the assignments $\dev[C]$ and $\dev[D]$,
    \[
    \begin{tikzcd}
      E_W \ar[r,"\ef{H}_E"] \ar[d,"\ef{G}_E"'] & E_D \ar[r,"\mathsf{dev}_D"] \ar[d,"R^{C,D}_E"] & \mathcal{P}(\dv{D}) \ar[d,"{\mathcal{P}(\mathsf{inr})}"] \\
      E_C \ar[d,"\mathsf{dev}_C"'] \ar[r,"L^{C,D}_E"'] & E_C +_{E_W} E_D \ar[rd,dashed,"\mathsf{dev}_{C \odot D}" description] & \mathcal{P}(\dv{C} + \dv{D}) \ar[d,equal] \\
      \mathcal{P}(\dv{C}) \ar[r,"{\mathcal{P}(\mathsf{inl})}"'] & \mathcal{P}(\dv{C} + \dv{D}) \ar[r,equal] & \mathcal{P}(\dv{C \odot D})
    \end{tikzcd}
  \]
  \item  $\efg{G \odot H}{W}{C \odot D}$ is identity-on-objects and $L_E^{C,D} \circ \ef{G}_E ~ (= R_E^{C,D} \circ \ef{H}_E)$ on arrows.%
  \end{itemize}
\end{construction}

\begin{example} \label{ex:printer}
  A simple example is given by our printer example from \Cref{sec:intro}, reproduced in \Cref{fig:commtensorex} (left), with the chosen monoidal graph of pure actions $W$ made explicit: in this case, it comprises the only pure action, namely the ``document'' generator. The commuting tensor product of this "effectful graph" with itself is the "effectful graph" in \Cref{fig:commtensorex} (right), in which we have two distinct printers.
  \begin{figure}[H]
    \includegraphics[width=\textwidth]{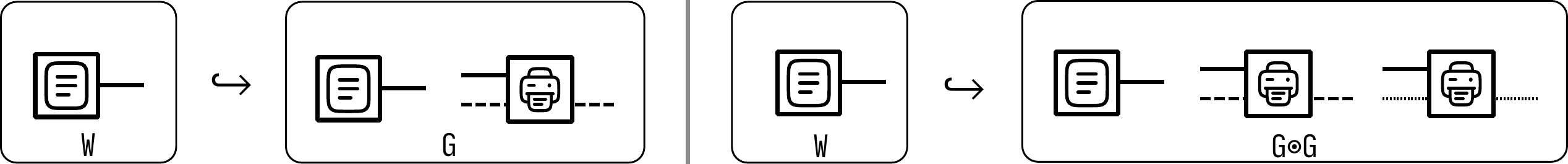}
    \caption{(Left) Effectful graph corresponding to the simple ``theory of a printer''. (Right) Commuting tensor product of this effectful graph with itself: the theory of two printers that may be used in parallel.}
    \label{fig:commtensorex}
  \end{figure}
\end{example}

Our next goal is to show that the free "effectful category" over a commuting tensor product of effectful graphs is isomorphic to the commuting tensor product of the free "effectful categories" over those graphs. First, we recall the definition of commuting tensor product of "effectful categories" given by Garner and López Franco \cite[\S 9.4]{commutativity}, see also \cite[Proposition 4.6.4]{earnshaw2025}. The definition builds on the auxiliary notion of commuting cospan.

\begin{definition}[Commuting cospan]
  A cospan of "effectful categories" \[F : (\efC{A} : \V \shortrightarrow \mathbb{A}) \to (\efC{K} : \V \shortrightarrow \mathbb{K}) \leftarrow (\efC{B} : \V \shortrightarrow \mathbb{B}) : G,\] over the same "monoidal category" $\V$, is said to be \emph{commuting} when, for every $f : A \to A'$ in $\mathbb{A}$ and $g : B \to B'$ in $\mathbb{B}$, the following diagram commutes in $\mathbb{K}$:
  \[
  \begin{tikzcd}
    A \otimes B \ar[d,"A \ltimes G(g)"'] \ar[r,"F(f) \rtimes B"] & A' \otimes B \ar[d,"A' \ltimes G(g)"] \\
    A \otimes B' \ar[r,"F(f) \rtimes B'"'] & A' \otimes B'    
  \end{tikzcd}
  \]
\end{definition}

\begin{definition}
  Let $\efC{A}$ and $\efC{B}$ be "effectful categories" over a "monoidal category" $\V$. The \emph{commuting tensor product of $\efC{A}$ and $\efC{B}$}, denoted $\efC{A} \odot \efC{B}$, is defined to be the apex of the initial commuting cospan $L : \efC{A} \to \efC{A} \odot \efC{B} \leftarrow \efC{B} : R$ of "effectful categories" over $\V$. Explicitly, for any other commuting cospan, $F : \efC{A} \to \efC{K} \leftarrow \efC{B} : G$ of "effectful categories" over $\V$, there exists a unique "morphism of effectful categories" $\efC{A} \odot \efC{B} \to \efC{K}$ such that:
  \[
  \begin{tikzcd}
    & \efC{A} \odot \efC{B} \ar[d,dashed] \\
    \efC{A} \ar[ur,"L"] \ar[r,"F"'] & \efC{K} & \efC{B} \ar[ul,"R"'] \ar[l,"G"] 
  \end{tikzcd}
  \]
\end{definition}

\begin{theorem} \label{thm:tensor}
  Given "effectful graphs" $\efg{G}{W}{C}$ and $\efg{H}{W}{D}$ over the same "monoidal graph" $W$, $\freeeff{(\ef{G} \odot \ef{H})} \cong \freeeff{\ef{G}} \odot \freeeff{\ef{H}}$.
\end{theorem}
\begin{proof}
  It suffices to show that $\freeeff[\ef{G} \odot \ef{H}]$ is the apex of the initial commuting cospan of "effectful categories" over the "free monoidal category" on $W$. In particular, we will show that:
  \[
  \begin{tikzcd}
    \freeeff{\ef{G}} \ar[r,"\freeeff{(L^{\ef{G},\ef{H}})}"]  & \freeeff{(\ef{G} \odot \ef{H})} & \freeeff{\ef{H}} \ar[l,"\freeeff{(R^{\ef{G},\ef{H}})}"']
  \end{tikzcd}
  \]
  is an initial commuting cospan, where $L^{\ef{G},\ef{H}} : \ef{G} \to \ef{G} \odot \ef{H}$ is the morphism of effectful graphs that is identity on objects and pure generators, and given by the left pushout injection $L^{C,D}_E$ on generators (Construction \ref{defn:ctg}), and similarly for $R^{\ef{G},\ef{H}} : \ef{H} \to \ef{G} \odot \ef{H}$.
  This cospan is commuting because of the way in which $\mathsf{dev}_{C \odot D}$ is constructed using the universal property of the pushout given by $L^{C,D}_E$ and $R^{C,D}_E$. In particular, we have by construction that for any $x \in E_C$ and $y \in E_D$, $L^{C,D}(x) \mathrel{\bot_{C \odot D}} R^{C,D}(y)$. It follows that $\freeeff(L^{\ef{G},\ef{H}})(f) \mathrel{\bot_{\freeeff(\ef{G} \odot \ef{H})}} \freeeff(R^{\ef{G},\ef{H}})(g)$ for any morphisms $f$ in $\freeeff(\ef{G})$ and $g$ in $\freeeff(\ef{H})$, from which it follows in turn that the cospan in question is commuting.

  To see that it is initial, suppose that we have another commuting cospan of effectful categories over $\freemon{W}$, call it $\freeeff{\ef{G}} \xrightarrow[]{P} (\efC{K} : \freemon{W} \shortrightarrow \mathcal{K}) \xleftarrow[]{Q} \freeeff{\ef{H}}.$  We obtain a function $(P \odot Q)_E : E_{C \odot D} \to E_{\mathcal{U}(\ef{K})}$ by the universal property of the pushout:
  \[
  \begin{tikzcd}
    E_W \ar[r,"\ef{H}_E"] \ar[d,"\ef{G}_E"'] & E_D \ar[r,"(\eta_D)_E"] \ar[d,"R^{C,D}_E"] & E_{\mathcal{U}(\mathcal{F}(D))} \ar[dd,"\mathcal{U}(Q)_E"] \\
    E_C \ar[d,"(\eta_C)_E"'] \ar[r,"L^{C,D}_E"'] & E_{C \odot D} \ar[rd,"(P \odot Q)_E"', dashed] \\
    E_{\mathcal{U}(\mathcal{F}(C))} \ar[rr,"\mathcal{U}(P)_E"'] && E_{\mathcal{U}({\ef{K})}}
  \end{tikzcd}
\]
  This mapping on generators of the device graphs defines an identity-on-objects "morphism of effectful graphs" $P \odot Q : \ef{G} \odot \ef{H} \to \mathscr{U}(\efC{K})$, where the fact that $P$ and $Q$ are a commuting cospan ensures that $f \mathrel{\bot_{G \odot H}} g \Rightarrow (P \odot Q)(f) \mathrel{\bot_{\mathscr{U}(\efC{K})}} (P \odot Q)(g)$. Using the bijection of hom-sets given by \Cref{thm:adjeff}, we obtain a morphism $(P \odot Q)^\# : \freeeff{(\ef{G} \odot \ef{H})} \to \efC{K}$ of $\mathsf{EffCat}$ such that:
    \[\begin{tikzcd}
    \freeeff{G} \ar[rdd,"P"'] \ar[r,"\freeeff{(L^{G,H})}"] & \freeeff{(\ef{G} \odot \ef{H})} \ar[dd,"(P\odot Q)^\#" description]  & \freeeff{H} \ar[l,"\freeeff{(R^{G,H})}"'] \ar[ddl,"Q"] \\\\
    & \efC{K}
    \end{tikzcd}\]
  Any other such morphism, transported back across the adjunction, induces another filler for our pushout diagram, and must therefore be equal to $(P \odot Q)^\#$.
\end{proof}

\section{Conclusion and further work}

We have shown how morphisms of free effectful categories can been seen as generalizations of Mazurkiewicz traces, by introducing a new notion of generating data for effectful categories, namely "effectful graphs". This point of view also leads to a simple presentation of the commuting tensor product of free effectful categories. In this paper, we have not given a treatment of presentations involving equations but this is would be a natural next step, and would allow us to express richer descriptions of concurrent systems, such as commuting tensor products of global state \cite{roman2025}.

Prior work has generalized Mazurkiewicz traces along different lines, such as to contextual traces \cite{choffrut2010contextual} or semicommutations \cite{clerbout1987semi}, so these might be investigated in light of this paper. For example, semicommutations generalize Mazurkiewicz traces by allowing commutations to be \emph{directed}, and this may correspond to a notion of order-\emph{enriched} "effectful category".

Since effectful categories are equivalent to strong promonads \cite{EPTCS380.20,commutativity}, recasting our concepts from that point of view would bring our work adjacent to the literature on combining monads \cite{hyland2006combining}, whose relation to the commuting tensor product should be investigated. This may also suggest appropriate generalizations. For example, we might be able to see device information as corresponding to a grading of a strong promonad.

String diagrams bearing some resemblance to resourceful traces have been used informally by Melliès in a treatment of the local state monad \cite{mellies2014local}. The local state monad arises as a combination of global state monads, linked by a theory of allocation: connections to our commuting tensor product should be investigated. Similarly, Barrett, Heijltjes, and McCusker \cite{barrett_et_al:LIPIcs.CSL.2023.10} have made informal use of string diagrams bearing resemblance to resourceful traces, in which devices stand for distinct effects in an effectful $\lambda$-calculus. These instances suggest that it would be worthwhile to formalize the diagrammatic point of view. Sobociński and the third author \cite{roman2025} proved that string diagrams over effectful graphs with a single global device do indeed give free premonoidal categories. However, since it is possible to construct effectful categories whose interference graphs contain an infinite number of maximal cliques, making the string-diagrammatic construction formal in our setting would require restricting to an appropriate subcategory of finitary effectful graphs. With this restriction, we see no reason why string diagrams for effectful categories should not extend to our setting.

\bibliography{bibliography}
\appendix
\newpage
\section{Omitted details} \label{app}

\interchange*
\begin{proof}
  By structural induction on $f$ and $g$. The cases where one of $f$ and $g$ are an identity or both morphisms are whiskered generators are straightforward by the equations of premonoidal categories, or by the imposed interchange equations for whiskered generators. Let $f = f_1 {\comp} f_2$, then by assumption we obtain $\dev[\mathcal{F}G](f_1) \cap \dev[\mathcal{F}G](g) = \dev[\mathcal{F}G](f_2) \cap \dev[\mathcal{F}G](g) = \varnothing$, and by hypothesis $f_1 "\parallel" g$ and $f_2 "\parallel" g$. Say $f_1 : A \to B, f_2 :  B \to C, g : D \to E$. We need to show
  $((f_1 \comp f_2) \rtimes D) {\comp} (C \ltimes g) = (A \ltimes g) {\comp} ((f_1 {\comp} f_2) \rtimes E)$. By the functoriality of $\rtimes$ and associativity of $\comp$, the left-hand side is equal to $(f_1 \rtimes D) \comp ((f_2 \rtimes D)\comp(C \ltimes g))$. Using the two induction hypotheses, this becomes $(A \ltimes g) \comp ((f_1 \rtimes E) \comp (f_2 \rtimes E))$ which is finally equal to the right-hand side by the functoriality of $\rtimes$. The argument for $g || f$ is symmetric, as well as the case when $g$ is a composite.
\end{proof}

\begin{construction}[""Free strict monoidal category""] \label{defn:freemon}
  For a "monoidal graph" $G$, we construct a "strict monoidal category" $\freemon{G}$ on $G$ as follows. The underlying category of $\freemon{G}$ has:
  \begin{itemize}
  \item set of objects $V_G^*$, the free monoid on $V_G$,
  \item set of morphisms constructed according to the following rules:
    \begin{mathpar}
      \inferrule[gen]{f : x \to y \in G}{f : x \to y}

      \inferrule[id]{A \in V_G^*}{1_A : A \to A}

      \inferrule[comp]{f : A \to B \\ g : B \to C}{fg : A \to C}

      \inferrule[tensor]{f_1 : A_1 \to B_1 \\ f_2 :  A_2 \to B_2}{f_1 \otimes f_2 : A_1 \otimes A_2 \to B_1 \otimes B_2}
    \end{mathpar}
    subject to the equations,
    \[
      (fg)h = f(gh) \\ 
       1_Af = f = f1_B \\
      (1_A \otimes 1_B) = 1_{A \otimes B}\]
      \[(f_1 \otimes f_2)(g_1 \otimes g_2) = (f_1g_1) \otimes (f_2g_2) \\
      (f \otimes g) \otimes h = f \otimes (g \otimes h)
    \]
  \item \textbf{composition} and \textbf{identities} are given by the corresponding inference rules.
  \end{itemize}
  The first two equations ensure that this data forms a category and the remaining equations ensure that the evident monoidal product is functorial.

  Similarly, the free "strict monoidal functor" on a "morphism of monoidal graphs" is given by the evident inductive extension.
\end{construction}

\begin{proposition}
  The construction of $\freemon{}$ defines a functor $\freemon{} : \mathsf{MonGraph} \to \mathsf{MonCat}$.
\end{proposition}

\begin{definition} \label{defn:umon}
  The ""underlying monoidal graph"" $"\mathcal{U}_{\smallotimes}"M$ of a "monoidal category" $M$ is defined to have set of objects the set of objects of $M$, and set of arrows,

  \[\coprod_{\substack{(X_1 \in \obj{M}, ..., X_n \in \obj{M}) \\ (Y_1 \in \obj{M}, ..., Y_m \in \obj{M})}} M(X_1 \otimes ... \otimes X_n; Y_1 \otimes ... \otimes Y_m)\]

  with the obvious assignments of sources and targets
 of arrows. For a "strict monoidal functor" $F : M \to N$, its ""underlying morphism of polygraphs"" $\mathcal{U}_{\smallotimes}(F)$ is given by the action of $F$ on objects and arrows.
 \end{definition}

\begin{proposition}
  \Cref{defn:umon} defines a functor $\mathcal{U}_{\smallotimes} : \MonCat \to \MonGraph$.
\end{proposition}

\begin{proposition}
  There is an adjunction $\freemon{} \dashv \mathcal{U}_{\smallotimes} : \mathsf{MonCat} \leftrightarrows \mathsf{MonGraph}$.
\end{proposition}
\begin{proof}
  A left adjoint to $\mathcal{U}_{\smallotimes}$ is constructed via string diagrams by Joyal and Street \cite{joyal91}.
\end{proof}

 \adj* 
 We begin constructing the unit, $\eta : 1_{\Dev} \to \mathcal{U} \circ \mathcal{F}$. A component $\eta_G : G \to \mathcal{U}(\mathcal{F}(G))$ acts on objects by inclusion as a singleton list, which we again denote simply as $A$,
 \[
  (\eta_G)_V : V_G \to V_{\mathcal{U}(\mathcal{F}(G))} : A \mapsto A.
\]

On arrows, recall that by definition, \[E_{\mathcal{U}(\mathcal{F}(G))} = \coprod_{\substack{(X_1 \in \obj{\mathcal{F}(G)}, ... , X_n \in \obj{\mathcal{F}(G)}) \\ (Y_1 \in \obj{\mathcal{F}(G)}, ..., Y_m \in \obj{\mathcal{F}(G)})}} \mathcal{F}G(X_1 \otimes ... \otimes X_n; Y_1 \otimes ... \otimes Y_m),\]
where $\otimes$ is list concatenation.
Let $A_i, B_j$ be elements of $V_G$, and $f : A_1, ..., A_n \to B_1, ..., B_m$ a morphism in $\mathcal{F}(G)$. Let us denote the coproduct injection into the component $(A_1, ..., A_n), (B_1, ..., B_m)$ of $E_{\mathcal{U}(\mathcal{F}(G))}$ by $\mathsf{in}$. Then we define,
\[(\eta_G)_E : E_G \to E_{\mathcal{U}(\mathcal{F}(G))} : f \mapsto {\mathsf{in}} \triangleright f \triangleleft,\]
where $\triangleright f \triangleleft$ denotes the primitive left-right whiskering of $f$ by the empty list, as in \Cref{defn:freepre}.
It is easy to check that this is a morphism of the underlying "monoidal graphs".
For the preservation of "orthogonality", if $f \perp_{G} g$ then by definition $\triangleright f \triangleleft \parallel \triangleright g \triangleleft$ and $\triangleright g \triangleleft \parallel \triangleright f \triangleleft$ in $\mathcal{F}(G)$. Therefore, $\triangleright f \triangleleft$ and $\triangleright g \triangleleft$ are not connected in the interference graph of $\mathcal{F}(G)$. Now, two maps in $\mathcal{U}(\mathcal{F}(G))$ will be orthogonal as long as they do not have a device in common, but they were to share a device, they would belong to a common maximal clique and thus be connected. This is a contradiction, so they must be orthogonal, $\eta_G(f) \mathbin{"\perp"_{\mathcal{U}(\mathcal{F}(G))}} \eta_G(g)$.

It remains to show that $\eta$ is natural. We must show that for any $\alpha : G \to H$ in $\Dev$ we have:
\[
\begin{tikzcd}
  G \ar[d,"\alpha"'] \ar[r,"\eta_G"] & \mathcal{U}(\mathcal{F}(G)) \ar[d,"\mathcal{U}(\mathcal{F}(\alpha))"]\\
  H \ar[r,"\eta_H"'] & \mathcal{U}(\mathcal{F}(H))
\end{tikzcd}
\]
We proceed by showing that the two above composite "device graph morphisms" have equal object mappings, and equal arrow mappings. First note $\mathcal{U}(\mathcal{F}(\alpha))_V = \alpha_V^{*}$ by definition, so we have,
\begin{align*}
   &((\eta_H)_V \circ \alpha_V)A
   = \alpha_V(A) = \alpha_V^{*}(A)\\
  &= (\alpha_V^{*}\circ(\eta_G)_V)A
  = (\mathcal{U}(\mathcal{F}(\alpha))_V \circ (\eta_G)_V)A
 \end{align*}
for all $A \in V_G$, as required. For the arrows, recall that by \Cref{lem:freepf}, $\mathcal{F}(\alpha)(\triangleright f \triangleleft) = \triangleright \alpha_E(f) \triangleleft$, which gives:
\begin{align*}
  &((\eta_H)_E\circ \alpha_E)f
    = (\eta_H)_E(\alpha_E(f))
   = \mathsf{in}~{\triangleright} {\alpha_E(f)} {\triangleleft}
  = \mathsf{in}~F(\alpha)(\triangleright f \triangleleft) \\
  &= \mathcal{U}(F(\alpha))_E(\mathsf{in}~{\triangleright} {f} {\triangleleft})
  = \mathcal{U}(F(\alpha))_E(\eta_G f)
  = (\mathcal{U}(F(\alpha))_E \circ (\eta_G)) f
\end{align*}
establishing the naturality of $\eta : 1_\Dev \to \mathcal{U} \circ \mathcal{F}$.

We now turn our attention to the counit $\varepsilon : \mathcal{F} \circ \mathcal{U} \to 1_{\PreMon}$, where $\varepsilon_\C : \mathcal{F}(\mathcal{U}(\C)) \to \C$ has action on objects $(\varepsilon_\C)(A_1,\ldots,A_n) := A_1\otimes \cdots \otimes A_n$, where $\otimes$ is the monoidal operation on objects that exists in a "premonoidal category". The action on morphisms is defined by structural recursion as follows: %
\[
\varepsilon_\C(f) =
\begin{cases}
  1_{\varepsilon_\C(A)} & \text{ if $f = 1_A$}\\
  \varepsilon_\C(g) {\comp} \varepsilon_\C(h) & \text{ if $f = g {\comp} h$}\\
  \varepsilon_\C(X) \triangleright f' \triangleleft \varepsilon_\C(Y) & \text{ if $f = X \triangleright \mathsf{in}  f' \triangleleft Y,$}\\
\end{cases}
\]
where $\mathsf{in}$ denotes a coproduct injection.
It is straightforward to check that this is well-defined, is a monoid morphism on objects, and preserves whiskerings.
It remains to show that $\varepsilon$ is natural. We must show that for any $F : \C \to \D$ in $\PreMon$ we have:
\[\begin{tikzcd}
    \mathcal{F}(\mathcal{U}(\C)) \ar[d,"\mathcal{F}(\mathcal{U}(F))"'] \ar[r,"\varepsilon_\C"] & \C \ar[d,"F"]\\
    \mathcal{F}(\mathcal{U}(\D)) \ar[r,"\varepsilon_\D"'] & \D
  \end{tikzcd}\]
We proceed by showing that the two composite functors have equal object mappings and equal arrow mappings. For the mapping on objects we have,
\begin{align*}
  & (F \circ \varepsilon_{\C})(A_1, \ldots, A_n) = F(A_1 \otimes \ldots \otimes A_n) \\
  &= F(A_1) \otimes \ldots \otimes F(A_n) = (\varepsilon_{\D} \circ \mathcal{F}(\mathcal{U}(F)))(A_1, \ldots, A_n).
\end{align*}
For the arrow mappings, we proceed by induction on $f$. The interesting case is when $f = X \triangleright \mathsf{in} f' \triangleleft Y$, for $\mathsf{in} f'$ an arrow of $\mathcal{U}(\C)$,
\begin{align*}
  & (F \circ \varepsilon_\C)(X \triangleright \mathsf{in} f' \triangleleft Y)
  = F(\varepsilon_\C(X) \triangleright \mathsf{in} f' \triangleleft \varepsilon_\C(Y))
  = F(\varepsilon_\C(X)) \triangleright F(\mathsf{in} f') \triangleleft F(\varepsilon_\C(Y)) \\
  &= \varepsilon_\D(\mathcal{F}(\mathcal{U}(F))(X)) \triangleright F(\mathsf{in} f') \triangleleft \varepsilon_\D(\mathcal{F}(\mathcal{U}(F))(Y))
  \\&= \varepsilon_\D(\mathcal{F}(\mathcal{U}(F))(X \triangleright \mathsf{in} f' \triangleleft Y))
  = (\varepsilon_\D \circ \mathcal{F}(\mathcal{U}(F)))(X \triangleright \mathsf{in} f' \triangleleft Y).
\end{align*}
The other cases follow directly from the inductive hypothesis together with the fact that each composite is a "strict premonoidal functor". It follows that $\varepsilon : \mathcal{F} \circ \mathcal{U} \to 1_{\PreMon}$ is natural.

Finally, to obtain an adjunction we must show the triangle identities:
  \[\begin{tikzcd}
    \mathcal{U} \ar[rd,"1_{\mathcal{U}}"'] \ar[r,"\mathcal{U}\eta"] & \mathcal{U} \circ \mathcal{F} \circ \mathcal{U} \ar[d,"\varepsilon\mathcal{U}"] \\
    & \mathcal{U}
  \end{tikzcd}\qquad
  \begin{tikzcd}
    \mathcal{F} \ar[rd,"1_{\mathcal{F}}"'] \ar[r,"\eta\mathcal{F}"] & \mathcal{F} \circ \mathcal{U} \circ \mathcal{F} \ar[d,"\mathcal{F}\varepsilon"] \\
    & \mathcal{F}
  \end{tikzcd}\]

For the first triangle, we have:
\begin{align*}
  & (\varepsilon\mathcal{U} \circ \mathcal{U}\eta)_\C(f)
  = \varepsilon\mathcal{U}_\C({\mathcal{U}\eta_\C}(f))
  = \mathcal{U}(\varepsilon_\C)(f)
  = f
  = 1_{\mathcal{U}(\C)}(f),
\end{align*}

from which we may conclude that the triangle in question commutes. For the second triangle, we proceed by induction on $f$.  The interesting case is when $f$ is a whiskered generating arrow, $f = X \triangleright \mathsf{in} f' \triangleleft Y$,
\begin{align*}
  & (\mathcal{F}\varepsilon \circ \eta\mathcal{F})_G(X \triangleright \mathsf{in} f' \triangleleft Y)
  = {\mathcal{F}\varepsilon_G}(\eta\mathcal{F}_G(X \triangleright \mathsf{in} f' \triangleleft Y)) \\
  &= {\mathcal{F}\varepsilon_G}(\mathcal{F}(\eta_G)(X \triangleright \mathsf{in} f' \triangleleft Y)) \\
  &= {\mathcal{F}\varepsilon_G}(\mathcal{F}(\eta_G)(X) \triangleright \mathcal{F}(\eta_G)(\mathsf{in} f') \triangleleft \mathcal{F}(\eta_G)(Y))) \\
  &= {\mathcal{F}\varepsilon_G}([X] \triangleright \eta_G(\mathsf{in} f') \triangleleft [Y]) \\
  &= {\varepsilon_{\mathcal{F}G}}([X]) \triangleright \mathsf{in} f' \triangleleft \varepsilon_{\mathcal{F}G}([Y]) \\
  &= X \triangleright \mathsf{in} f' \triangleleft Y,
\end{align*}

where for $X = X_1,...,X_n$, we denote by $[X]$ the list $[X_1],...,[X_n]$.
The case where $f$ is $1_A$ follows from the (easily checked) fact that the claim holds for the object mapping, and the case of composition follows directly from the inductive hypothesis together with the fact that the composite is a "strict premonoidal functor".

\adjeff*
\begin{proof}
  We shall make use of the adjunction of \Cref{thm:adj} and the adjunction between "monoidal graphs" and "monoidal categories" (\cite{joyal91}). Let $\efg{E}{V}{G}$ be an "effectful graph". We first define the components of the unit, namely a "morphism of effectful graphs"
  \[\eta_{\ef{E}} : (\efg{E}{V}{G}) \to (\mathscr{U}\circ\mathscr{F})(\efg{E}{V}{G}),\]
by letting the morphism between the "monoidal graphs" be the unit of the adjunction between "monoidal graphs" and "monoidal categories", $\eta^{\smallotimes}_V : V \to \mathcal{U}_{\smallotimes} \mathcal{F}_{\smallotimes} V$, and the morphism between the "device graphs" be given by the unit $\eta_G : G \to \mathcal{U}\mathcal{F} G$ of the adjunction constructed in \Cref{thm:adj}. It is straightforward to verify that the required square commutes.

  Let $\efc{E}{\V}{\C}$ be an "effectful category". For the components of the counit, we define  a "morphism of effectful categories",
  \[\varepsilon_{\efC{E}} : (\mathscr{F} \circ \mathscr{U})\efC{E} \to \efC{E}\]
  by letting the component between the "monoidal categories" be given by the counit of the adjunction between "monoidal graphs" and "monoidal categories", $\varepsilon^{\smallotimes}_{\C} : \mathcal{F}_{\smallotimes}\mathcal{U}_{\smallotimes}\C \to \C$, and the morphism between "premonoidal categories" be given by the counit of the adjunction constructed in \Cref{thm:adj}.
  
  It remains to show that the unit and counit satisfy the triangle identities, which follows simply from the triangle identities of the two adjunctions involved.
\end{proof}

\end{document}